\documentclass[sigconf]{acmart}

\AtBeginDocument{%
  \providecommand\BibTeX{{%
    \normalfont B\kern-0.5em{\scshape i\kern-0.25em b}\kern-0.8em\TeX}}}

\copyrightyear{2022} 
\acmYear{2022} 
\setcopyright{acmcopyright}\acmConference[WWW '22]{Proceedings of the ACM Web Conference 2022}{April 25--29, 2022}{Virtual Event, Lyon, France}
\acmBooktitle{Proceedings of the ACM Web Conference 2022 (WWW '22), April 25--29, 2022, Virtual Event, Lyon, France}
\acmPrice{15.00}
\acmDOI{10.1145/3485447.3512010}
\acmISBN{978-1-4503-9096-5/22/04}

\usepackage[ruled,noend,linesnumbered]{algorithm2e}
\usepackage{multirow}
\usepackage{array}
\usepackage{graphicx}
\usepackage{subcaption}
\usepackage{amsthm}
\newtheorem{definition}{Definition}[section]
\newtheorem{proposition}{Proposition}[section]

\graphicspath{ {./figures/} }

\newcommand{\ie}{\emph{i.e., }}
\newcommand{\eg}{\emph{e.g., }}

\newcommand{\wrt}{\emph{w.r.t. }}

\newcommand{\Mat}[1]{\mathbf{#1}}
\newcommand{\Set}[1]{\mathcal{#1}}

\begin{document}

\title{Cross Pairwise Ranking for Unbiased Item Recommendation}

\author{Qi Wan}
\email{wqq17@mail.ustc.edu.cn}
\orcid{0000-0003-4065-0756}
\affiliation{%
  \institution{University of Science and Technology of China}
  \country{China}
}

\author{Xiangnan He}
\authornote{Xiangnan He is the corresponding author.}
\email{xiangnanhe@gmail.com}
\affiliation{%
  \institution{University of Science and Technology of China}
  \country{China}
}

\author{Xiang Wang}
\email{xiangwang1223@gmail.com}
\affiliation{%
  \institution{University of Science and Technology of China}
  \country{China}
}

\author{Jiancan Wu}
\email{wjc1994@mail.ustc.edu.cn}
\affiliation{%
  \institution{University of Science and Technology of China}
  \country{China}
}

\author{Wei Guo}
\email{guowei67@huawei.com}
\affiliation{%
  \institution{Huawei Noah’s Ark Lab}
  \country{China}
}

\author{Ruiming Tang}
\email{tangruiming@huawei.com}
\affiliation{%
  \institution{Huawei Noah’s Ark Lab}
  \country{China}
}

\renewcommand{\shortauthors}{Trovato and Tobin, et al.}

\begin{abstract}
  Most recommender systems optimize the model on observed interaction data, which is affected by the previous exposure mechanism and exhibits many biases like popularity bias.
  The loss functions, such as the mostly used pointwise Binary Cross-Entropy and pairwise Bayesian Personalized Ranking, are not designed to consider the biases in observed data. As a result, the model optimized on the loss would inherit the data biases, or even worse, amplify the biases. For example, a few popular items take up more and more exposure opportunities, severely hurting the recommendation quality on niche items --- known as the notorious Mathew effect.

  In this work, we develop a new learning paradigm named \textit{Cross Pairwise Ranking} (CPR) that achieves unbiased recommendation without knowing the exposure mechanism.
  Distinct from inverse propensity scoring (IPS), we change the loss term of a sample --- we innovatively sample multiple observed interactions once and form the loss as the combination of their predictions.
  We prove in theory that this way offsets the influence of user/item propensity on the learning, removing the influence of data biases caused by the exposure mechanism.
  Advantageous to IPS, our proposed CPR ensures unbiased learning for each training instance without the need of setting the propensity scores.
  Experimental results demonstrate the superiority of CPR over state-of-the-art debiasing solutions in both model generalization and training efficiency.
  The codes are available at \url{https://github.com/Qcactus/CPR}.
\end{abstract}

\ccsdesc[500]{Information systems~Retrieval models and ranking}

\ccsdesc[500]{Information systems~Recommender systems}

\keywords{Unbiased Learning-to-Rank, Popularity Bias, Recommendation}

\maketitle

\section{Introduction}

Recommender systems have become the core of many online platforms, including e-Commerce, streaming media, social networks etc.
Existing models mostly follow a supervised learning paradigm, which treats historical interactions (\eg implicit feedback \cite{bpr,neumf,lightgcn} like clicks and purchases) as the labeled data, and learns user-item relevance by fitting the labeled data.
Two types of loss functions are intensively used to optimize the model parameters:
(1) pointwise loss, which captures a user's preference on single item by minimizing the discrepancy between the target and predicted relevance score of a user-item pair, \eg Binary Cross-Entropy (BCE)~\cite{neumf,wide_deep,logistic_mf} and Mean Squared Error (MSE)~\cite{wmf,fast_mf};
and (2) pairwise loss, which models a user's preference on two items by encouraging the prediction of positive item to be higher than that of the negative item, \eg Bayesian Personalized Ranking (BPR)~\cite{bpr}.

However, such standard loss functions are prone to the biases inherent in the observed interaction data.
For example, the interaction data usually presents a long-tail distribution on item popularity \cite{multi-sided}, that is, a small number of popular items occupy most interactions.
Recent studies show that item popularity confounds user true preference and observed interactions \cite{pda,macr} and should be properly handled for quality recommendation.
Nonetheless, mainstream loss functions are designed to recover the historical data --- assuming that observed interactions reflect user preference faithfully --- while leaving the bias effect untouched.
Hence, when building recommender systems with these losses, a few popular items take up more and more exposure opportunities, severely hurting the recommendation quality on niche items.

How to mitigate the bias and perform unbiased estimation of user preference has become one central theme in recommendation \cite{bias_debias}.
One prevalent solution is Inverse Propensity Scoring (IPS)~\cite{relmf,ubpr,uebpr,du}, which reweights each data sample by the inverse of its propensity score (i.e., the exposure probability). Despite that IPS is unbiased in theory, it suffers from practical limitations:
(1) it is challenging to accurately estimate the propensity score of each sample, since the exposure mechanism is seldom known \cite{rec_treat}; and (2) the reweighted loss usually exhibits a high variance especially for implicit feedback~\cite{relmf,ubpr,ips_norm}, which implies that the losses of individual samples fluctuate heavily from their expected values.

In this work, we propose a sample-level unbiased learning paradigm, Cross Pairwise Ranking (CPR), without the need of quantifying the exposure mechanism.
The key assumption is that the propensity score can be factorized into user propensity, item propensity, and user-item relevance. With this assumption, we can offset the impact of the user/item propensity by designing a novel loss function that combines purposefully the chosen user-item pairs. Briefly speaking, superior to IPS, the proposed CPR ensures unbiased learning for each individual sample and suits better for debiasing the learning from implicit feedback. Furthermore, for faster convergence in training, we borrow the idea of hard negative sampling \cite{dns,improve_sampling} and devise a dynamic sampling method for CPR.

In a nutshell, we summarize the contributions as follows:
\begin{itemize}
  \item We analyze the commonly used pointwise and pairwise loss functions in recommender system from a new perspective, showing that they are biased in deriving the correct ranking of user preference.
  \item We propose a new learning method CPR and offer theoretical analyses of its unbiasedness in estimating user preference on individual samples.
  \item We conduct empirical studies on two benchmark datasets, validating the effectiveness of CPR in recommendation debiasing on multiple backbone models.
\end{itemize}
\section{Preliminaries}

We first introduce the basic notations. Then we formulate a new perspective on the unbiased learning and analyze the biasedness of the pointwise and pairwise loss from this perspective.

\subsection{Notations}

Let $u \in \Set{U}$ be a user and $i \in \Set{I}$ be an item. We call a user-item pair a positive pair if we observe an interaction between them, and call it a negative pair otherwise.

Following the notations in previous literature~\cite{relmf,ubpr}, we define three binary variables: the interaction variable $Y_{u,i}$, the relevance variable $R_{u,i}$ and the observation variable $O_{u,i}$.
$Y_{u,i}=1$ represents an interaction happens between user $u$ and item $i$, and $0$ otherwise.
$R_{u,i}=1$ means the user likes the item, and $0$ otherwise.
$O_{u,i}=1$ means the user observes the item, and $0$ otherwise.
An interaction happens when the user likes and observes the item, \ie $Y_{u,i}=R_{u,i} \cdot O_{u,i}$, leading to:
\begin{equation}\label{basic_prob}
    \begin{split}
        P(Y_{u,i}=1)&=P(R_{u,i}=1,O_{u,i}=1)\\
        &=P(R_{u,i}=1)P(O_{u,i}=1 \mid R_{u,i}=1).
    \end{split}
\end{equation}
$P(R_{u,i}=1)$ indicates the relevance probability between user $u$ and item $i$.
$P(O_{u,i}=1 \mid R_{u,i}=1)$ is the exposure probability --- the probability of item $i$ being observed by user $u$ conditioned on their relevance. Here we do not assume the independence between $R_{u,i}$ and $O_{u,i}$ as in previous studies~\cite{relmf,ubpr}, since in the real world, users are more likely to be exposed to highly relevant items because of the exposure mechanism of the previous recommender that also captures user-item relevance.
Besides, due to the bias in the previous exposure mechanism, among all potentially relevant items, users are more likely to be exposed to the popular ones, \ie $P(O_{u,i}=1 \mid R_{u,i}=1)$ of popular items is generally higher than that of niche items, which increases their interaction probability $P(Y_{u,i}=1)$ and makes the recommendation lean towards them.

Let ${s}_{u,i}$ be $\ln P(R_{u,i}=1)$. We call ${s}_{u,i}$ the true relevance score between user $u$ and item $i$. More precisely, ${s}_{u,i}=\ln \left[ a \cdot P(R_{u,i}=1) \right]$, where $a$ is a positive constant that changes the scale of ${s}_{u,i}$; we ignore this constant since it does not matter in the following analyses. We use $\hat{s}_{u,i}$ to denote the predicted relevance score, which is typically obtained by feeding the user and item embeddings into an interaction function like inner product or a neural network \cite{survey_neurec}. $\hat{s}_{u,i}$ is expected to estimate ${s}_{u,i}$, but the challenge is that we only have observation data on $Y_{u, i}$, whereas $R_{u,i}$ and $O_{u,i}$ are unobservable.

\subsection{Definition of Unbiasedness}

Previous studies on IPS \cite{relmf,ubpr} define the unbiasedness of losses from the statistical view of expectation, making the loss expectation equal to the ideal loss by reweighting samples. However, it suffers from some practical limitations, such as the difficulty of setting propensity scores and the high variance of the reweighted loss. We leave the discussion to Section \ref{sec:ips}.

In this work, we propose a new view on the unbiasedness of loss, that is free from the limitations of IPS:
\begin{definition}\label{def:unbias}
    A loss function $\mathcal{L}$ is unbiased if it optimizes the ranking of predicted user-item relevance scores towards that of the true relevance scores:
    $$\hat{s}_{u,i}>\hat{s}_{u,j} \Leftrightarrow s_{u,i}>s_{u,j} \text{ when $\mathcal{L}$ converges.}$$
\end{definition}

In the rest of this section, we will briefly introduce the traditional pointwise and pairwise losses, and analyze their biasedness from this new perspective.

\subsection{Biasedness of Pointwise and Pairwise Loss}

To learn user preference from historical interactions, most models follow the supervised learning paradigm, which encourages the relevance scores to recover the labels.
Two standard loss functions are widely used to optimize the model parameters:
(1) The pointwise loss like BCE \cite{neumf,wide_deep,logistic_mf} aims to capture the user preference on single items via minimizing the prediction discrepancy:
\begin{equation}\label{bce_mse}
    \mathcal{L}_{BCE}=-\sum_{(u,i)\in \Set{D}} Y_{u,i} \ln \sigma(\hat{s}_{u,i}) + (1-Y_{u,i}) \ln (1-\sigma(\hat{s}_{u,i})),
\end{equation}
where $\Set{D}=\Set{D}^+ \cup \Set{D}^-$; $\Set{D}^+$ is the set of observed interactions, and $\Set{D}^-$ is the set of missing data, and $\sigma(\cdot)$ is the sigmoid function.
(2) The pairwise loss, such as BPR \cite{bpr}, models the user preference on two items.
It encourages the prediction of the positive item to be higher than that of the negative item:
\begin{equation}\label{bpr}
    \mathcal{L}_{BPR}=-\sum_{(u,i,j)\in \Set{D}_S} \ln \sigma(\hat{s}_{u,i}-\hat{s}_{u,j}),
\end{equation}
where $\Set{D}_S=\{(u,i,j) \mid Y_{u,i}=1,\ Y_{u,j}=0\}$.
Next, we analyze how they lead to the biased estimation of users' true preference.

\subsubsection{Biasedness of Pointwise Loss}

\begin{proposition}
    Pointwise loss is biased.
\end{proposition}

\begin{proof}

    The pointwise loss in Equation \eqref{bce_mse} performs binary classification between the relevance scores of positive pairs and negative pairs, which can be represented by the following inequality:
    \begin{equation}\label{point_actual}
        \begin{cases}
            \hat{s}_{u,i} \geq C & \text{if $Y_{u,i}=1$}, \\
            \hat{s}_{u,i} < C    & \text{if $Y_{u,i}=0$}.
        \end{cases}
    \end{equation}
    where $C$ is the constant that denotes the classification threshold.

    To maximize the likelihood of observing $Y_{u,i}$, we should rank $P(Y_{u,i}=1)$ as follows:
    \begin{equation*}
        \begin{cases}
            \ln P(Y_{u,i}=1) \geq C & \text{if $Y_{u,i}=1$}, \\
            \ln P(Y_{u,i}=1) < C    & \text{if $Y_{u,i}=0$}.
        \end{cases}
    \end{equation*}
    which can be rewritten as:
    \begin{equation}\label{raw_point_ideal}
        \begin{cases}
            {s}_{u,i} + \ln P(O_{u,i}=1 \mid R_{u,i}=1) \geq C & \text{if $Y_{u,i}=1$}, \\
            {s}_{u,i} + \ln P(O_{u,i}=1 \mid R_{u,i}=1) < C    & \text{if $Y_{u,i}=0$}.
        \end{cases}
    \end{equation}
    This is the expected ranking of ${s}_{u,i}$.
    Compared with this desired ranking, Equation \eqref{point_actual} actually uses $\hat{s}_{u,i}$ to model ${s}_{u,i} + \ln P(O_{u,i}=1 \mid R_{u,i}=1)$.
    The resultant $\hat{s}_{u,i}$ is generally higher for popular items, since their exposure probability tends to be higher. Therefore, the pointwise loss is biased, and models optimized with it will favor popular items.
\end{proof}

\subsubsection{Biasedness of Pairwise Loss}

\begin{proposition}
    Pairwise loss is biased.
\end{proposition}

\begin{proof}
    The pairwise loss in Equation \eqref{bpr} encourages the scores of positive pairs to be higher than those of negative pairs for each user, which can be represented by the following inequality:
    \begin{equation}\label{pair_actual}
        \hat{s}_{u,i} - \hat{s}_{u,j} > 0,\text{ if $Y_{u,i}=1,Y_{u,j}=0$}.
    \end{equation}

    According to the expected ranking \eqref{raw_point_ideal}, we have:
    \begin{equation*}
        \begin{split}
            &{s}_{u,i} + \ln P(O_{u,i}=1 \mid R_{u,i}=1) \geq C, \text{ if $Y_{u,i}=1$};\\
            &{s}_{u,j} + \ln P(O_{u,j}=1 \mid R_{u,j}=1) < C, \text{ if $Y_{u,j}=0$}.
        \end{split}
    \end{equation*}
    By a subtraction between these two inequalities, we can write this ranking in a pairwise form:
    \begin{multline}\label{raw_pair_ideal}
        {s}_{u,i}+\ln P(O_{u,i}=1 \mid R_{u,i}=1) - \left[{s}_{u,j}+\ln P(O_{u,j}=1 \mid R_{u,j}=1)\right]\\ > 0,
        \text{ if $Y_{u,i}=1,\ Y_{u,j}=0$}.
    \end{multline}
    Compared with this desired ranking, Equation \eqref{pair_actual} also uses $\hat{s}_{u,i}$ to model ${s}_{u,i} + \ln P(O_{u,i}=1 \mid R_{u,i}=1)$, which is generally higher for popular items, thus the pairwise loss is biased.
\end{proof}
\section{Method}

In this section, we propose a new loss named CPR, and theoretically prove its unbiasedness under a reasonable assumption. Then we extend this loss to a more general form for a stronger supervision on learning. Finally, we design a dynamic sampling algorithm for CPR to speed up the training and also improve its performance.

\subsection{CPR Loss}

Below we present the proposed CPR loss and discuss its rationale in the following subsections. We first construct a training sample by selecting two positive user-item pairs $(u_1, i_1)$ and $(u_2, i_2)$, so that $u_1\ne u_2$, meanwhile $(u_1, i_2)$ and $(u_2, i_1)$ are negative pairs.
On such training samples, we propose a novel loss as follows:
\begin{equation}
    \begin{split}
        \mathcal{L}_{CPR}=-\sum_{(u_1, u_2, i_1, i_2) \in \Set{D}_2} \ln \sigma \left[\frac{1}{2}(\hat{s}_{u_1,i_1}+\hat{s}_{u_2,i_2}-\hat{s}_{u_1,i_2}-\hat{s}_{u_2,i_1})\right]
    \end{split}
\end{equation}
where $\Set{D}_2=\{(u_1,u_2,i_1,i_2) \mid Y_{u_1, i_1}=1,Y_{u_2, i_2}=1,Y_{u_1, i_2}=0,Y_{u_2, i_1}=0\}$ denotes the training data.

\subsection{Unbiasedness of CPR Loss}

CPR loss encourages the sum score of the two positive pairs to be higher than that of the two negative pairs, which are composed by crossing the two pairs, \ie
\begin{multline}\label{cpr_2}
    \hat{s}_{u_1,i_1}+\hat{s}_{u_2,i_2}-\hat{s}_{u_1,i_2}-\hat{s}_{u_2,i_1}>0,\\
    \text{if $Y_{u_1,i_1}=1,Y_{u_2,i_2}=1,Y_{u_1,i_2}=0,Y_{u_2,i_1}=0$}.
\end{multline}
To demonstrate its unbiasedness, we first make a mild assumption: the exposure probability can be factorized into user propensity, item propensity, and user-item relevance, which is formulated as:
\begin{equation}\label{o_factorized}
    P(O_{u,i}=1 \mid R_{u,i}=1)=p_u \cdot p_i \cdot P(R_{u,i}=1)^\alpha,
\end{equation}
where $p_u$ and $p_i$ are user-specific and item-specific propensity, respectively, typically higher for active users and popular items;
$P(R_{u,i}=1)^\alpha$ reflects that higher relevance contributes more to the exposure probability;
$\alpha$ is a positive constant for smoothing purpose.
With this assumption, the expected ranking in Equation \eqref{raw_point_ideal} can be rewritten as:
\begin{equation}\label{point_ideal}
    \begin{cases}
        s_{u,i} \geq \frac{1}{1+\alpha}(C - \ln p_u - \ln p_i) & \text{if $Y_{u,i}=1$}, \\
        s_{u,i} < \frac{1}{1+\alpha}(C - \ln p_u - \ln p_i)    & \text{if $Y_{u,i}=0$}.
    \end{cases}
\end{equation}
We can revisit the pointwise ranking in Equation \eqref{point_actual}, which is biased against this correct ranking owing to the impact of $p_u$ and $p_i$.
Similarly, by rewriting the correct pairwise ranking in Equation \eqref{raw_pair_ideal}, we find that the pairwise ranking in Equation \eqref{pair_actual} is biased against this correct ranking due to the impact of $p_i$:
\begin{equation}\label{pair_ideal}
    s_{u,i} - s_{u,j} > \frac{1}{1+\alpha}(- \ln p_i + \ln p_j),\text{ if $Y_{u,i}=1,Y_{u,j}=0$}.
\end{equation}

\begin{proposition}
    CPR loss is unbiased under the assumption of Equation \eqref{o_factorized}.
\end{proposition}

\begin{proof}
    Given the expected ranking under the assumption, we inspect one training sample involving the positive pairs $(u_1,i_1)$ and $(u_2,i_2)$ and the negative pairs $(u_1,i_2)$ and $(u_2,i_1)$, and have:
    \begin{equation*}
        s_{u_1,i_1} \geq \frac{1}{1+\alpha}(C - \ln p_{u_1} - \ln p_{i_1}),\
        s_{u_2,i_2} \geq \frac{1}{1+\alpha}(C - \ln p_{u_2} - \ln p_{i_2}),
    \end{equation*}
    \begin{equation*}
        s_{u_1,i_2} < \frac{1}{1+\alpha}(C - \ln p_{u_1} - \ln p_{i_2}),\
        s_{u_2,i_1} < \frac{1}{1+\alpha}(C - \ln p_{u_2} - \ln p_{i_1}).
    \end{equation*}
    By a subtraction between the sum of the first two inequalities and the sum of the last two, we get:
    \begin{multline}
        s_{u_1,i_1}+s_{u_2,i_2}-s_{u_1,i_2}-s_{u_2,i_1}>0,\\
        \text{if $Y_{u_1,i_1}=1,\ Y_{u_2,i_2}=1,\ Y_{u_1,i_2}=0,\ Y_{u_2,i_1}=0$}.
    \end{multline}
    $p_u$ and $p_i$ are canceled here, since they have the same total impact on the positive pairs as on the negative pairs. Clearly, CPR ranking (Equation~(\ref{cpr_2})) coincides with this form of expected ranking--- $s_{u,i}$ is successfully modeled by $\hat{s}_{u,i}$. Thus, CPR loss is unbiased.
\end{proof}

\subsection{Extending to More Interactions}\label{sec:extend}
The CPR loss discussed above only uses two observed interactions as a training sample.
We now extend it to $k(k \ge 2)$ interactions, while preserving its unbiasedness:
\begin{multline}\label{cpr_loss}
    \mathcal{L}_{CPR}=-\sum_{k} \sum_{(u_1,\dots,u_k, i_1,\dots,i_k) \in \Set{D}_k}\\
    \ln \sigma
    \left[\frac{1}{k} \cdot (\hat{s}_{u_1,i_1}+\dots+\hat{s}_{u_k,i_k}
        -\hat{s}_{u_1,i_2}-\dots-\hat{s}_{u_k,i_1})\right]
\end{multline}
where $\Set{D}_k=\{(u_1,\dots,u_k,i_1,\dots,i_k) \mid Y_{u_1, i_1}=1,Y_{u_2, i_2}=1,\dots,\allowbreak Y_{u_k, i_k}=1, Y_{u_1, i_2}=0,Y_{u_2, i_3}=0,\dots,Y_{u_k, i_1}=0 \}$, $k \ge 2$, and $\frac{1}{k}$ scales the loss of samples with different numbers of interactions.
We omit the detailed proof of its unbiasedness, which is similar to the previous section.
Figure~\ref{fig:sample} illustrates the composition of samples with different numbers of interactions.

These different types of samples set different restrictions to the ranking, so it is a natural idea to set $k$ as many values in the loss. But in experiments, CPR usually achieves the best performance when $k=2,3$, and the performance cannot be improved when k takes more values greater than 3.
One possible reason is that a larger $k$ makes the training inflexible --- it forces more pairs to be trained together, even when some of them are about to overfit. In the experiments described later, samples are only drawn from $\Set{D}_2$ and $\Set{D}_3$.

The complete loss function used for CPR is:
\begin{equation}
    \mathcal{L}=\mathcal{L}_{CPR}+ \lambda (\|\Mat{\Theta}\|^2),
\end{equation}
where $\Mat{\Theta}$ denotes model parameters, $\lambda$ is the $L_2$ regularization coefficient to control overfitting.

\begin{figure}
    \centering
    \includegraphics[width=0.8\linewidth]{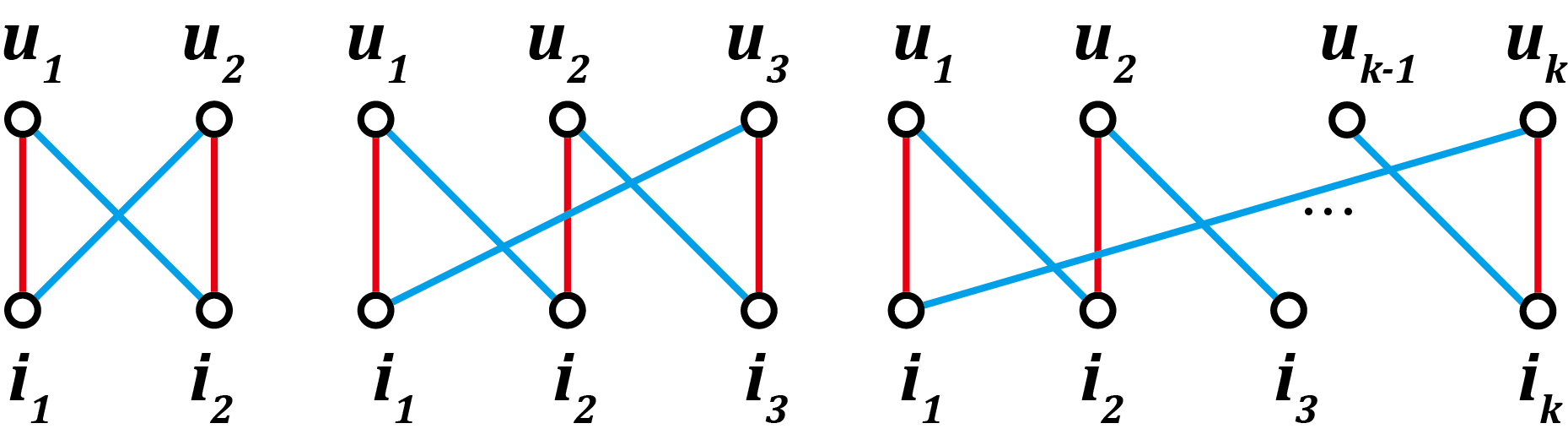}
    \vspace{-10pt}
    \caption{Composition of CPR samples. Red lines denote positive pairs. Blue lines denote negative pairs. Graphs from left to right denote the composition of  samples in $\Set{D}_2$, $\Set{D}_3$ and $\Set{D}_k$, respectively.}
    \label{fig:sample}
    \vspace{-10pt}
\end{figure}

\subsection{Dynamic Sampling}
Obviously, how to construct the sample set $\Set{D}_k$ plays a critical role in CPR.
One straightforward solution is to perform a random sampling: first draw $k$ non-overlap positive user-item pairs, then accept them as one sample if their cross-combinations are all negative pairs, otherwise discard them and redraw.
In this way, all samples are selected with the same probability.
However, during the training, some hard samples \cite{improve_sampling} might take more iterations to converge, while some easy samples have less contribution to the model optimization and should be avoided.

Inspired by DNS~\cite{dns}, which is a dynamic negative sampling strategy for BPR \cite{bpr}, we devise a dynamic sampling strategy for CPR, with the aim of assigning hard samples with higher sampling probability and training them more frequently.
Algorithm~\ref{algo:dys} details how to dynamically select a batch of samples.
We use $b$, $\beta(\beta \ge 1)$ and $\gamma(\gamma>1)$ to denote the batch size, the dynamic sampling rate and the choosing rate, respectively.
First we randomly select $b\beta\gamma$ samples, each of which contains $k$ observed interactions (line \ref{line:1}).
The choosing rate $\gamma$ is to increase the number of initial samples to ensure that after discarding inappropriate samples in the following step, we can still collect the desired number of samples.
Next, we discard the samples whose cross-combinations are not $k$ negative pairs, and obtain $b\beta$ useful samples (line \ref{line:2_start}-\ref{line:2_stop}).
We measure the difficulty of a sample by the value of $\frac{1}{k} \cdot (\hat{s}_{u_1,i_1}+\hat{s}_{u_2,i_2}+\dots+\hat{s}_{u_k,i_k}-\hat{s}_{u_1,i_2}-\hat{s}_{u_2,i_3}-\dots-\hat{s}_{u_k,i_1})$ --- the smaller the value is, the harder it is for this sample to reach its target ranking. Therefore, we select samples with the least values of this expression as a batch of samples (line \ref{line:3_start}-\ref{line:3_stop}).
\vspace{-10pt}

\begin{algorithm}
    \DontPrintSemicolon
    \caption{Dynamic Sampling}
    \label{algo:dys}
    \KwIn{Batch size $b$, dynamic sampling rate $\beta \ge 1$, choosing rate $\gamma > 1$, parameter $k$}
    \KwOut{A batch of samples $\Set{D}_k^{batch}$ for training}
    $\Set{S} \longleftarrow \emptyset$\;
    Randomly select $b \beta \gamma$ samples, each of which contains $k$ positive user-item pairs\;\label{line:1}
    \For{each sample $\{(u_1,i_1),(u_2,i_2),\dots,(u_k,i_k)\}$}{\label{line:2_start}
        \If{$\{(u_1,i_2),(u_2,i_3),\dots,(u_k,i_1)\}$ is a set of negative pairs}{
            $\Set{S} \longleftarrow \Set{S} \cup (u_1,\dots,u_k,i_1,\dots,i_k)$\;
            \lIf{$|\Set{S}|=b\beta$}{
                break
            }
        }
    }\label{line:2_stop}
    Calculate $O_j=\frac{1}{k} \cdot (\hat{s}_{u_1,i_1}+\hat{s}_{u_2,i_2}+\dots+\hat{s}_{u_k,i_k}-\hat{s}_{u_1,i_2}-\hat{s}_{u_2,i_3}-\dots-\hat{s}_{u_k,i_1})$ for $S_j=(u_1,\dots,u_k,i_1,\dots,i_k) \in \Set{S}$\;\label{line:3_start}
    Sort $S_1,\dots,S_{b\beta}$ by the ascending order of $O_1,\dots,O_{b\beta}$\;
    $\Set{D}_k^{batch} \longleftarrow$ the first b samples in the sorted list\;
    \Return{$\Set{D}_k^{batch}$}\label{line:3_stop}
\end{algorithm}
\vspace{-10pt}

\section{Discussion}

We compare CPR with related work to demonstrate the superiority of its novel design. First we compare with a prevalent debiasing method, IPS, and show how CPR avoids the difficulty of setting propensity scores and the high-variance problem. Then we discuss the difference between CPR and Setwise Ranking~\cite{cofiset,setrank,set2setrank}.

\subsection{Comparison with IPS}\label{sec:ips}

IPS and CPR are both debiasing approaches that focus on the loss function, but with different views on the unbiasedness of loss and hence provide different solutions. Here we briefly explain the methodology of IPS and how our CPR avoids or mitigates the inherent problems of IPS.

Let $\mathcal{L}_{ideal}$ be the ideal loss function where the exposure mechanism is properly handled. According to IPS theory, the loss function $\mathcal{L}$ is unbiased if its expectation across the probabilistic distribution of exposure is equal to the ideal loss:
$$\mathbb{E}_{O_{i,j}}(\mathcal{L}(\hat{s}_{u,i}))=\mathcal{L}_{ideal}(\hat{s}_{u,i})$$
With this definition, the authors prove that traditional losses are biased for having an expectation not equal to the ideal loss. To deal with this problem, they reweight the label $Y_{u,i}$ used in the traditional loss functions by dividing it by a \emph{propensity score}, which is an estimate of $P(O_{i,j}=1)$. However, two problems emerge from this solution:
\begin{itemize}
    \item The propensity score is typically calculated by a power-law function of the item degree (the number of its occurrences in the training set), which is an oversimplification of the exposure mechanism.
    \item Although the reweighted loss can have an ideal expectation (if we ignore the inaccuracy of propensity scores for now), its variance can be large (Theorem 4.4 in \cite{relmf}, Theorem 3.4 in \cite{ubpr}, and \cite{ips_norm}), in other words, for each sample, the reweighted loss deviates from the ideal loss. Variance reduction techniques are usually applied, such as clipping the weight or the loss of each sample to a certain range; but in this way, the loss expectation is no longer unbiased.
\end{itemize}

Our CPR alleviates these two problems:
\begin{itemize}
    \item CPR takes a more general assumption on the exposure mechanism, \ie Equation~(\ref{o_factorized}), without explicitly setting any propensity score.
    \item CPR changes the composition of loss instead of reweighting the original loss, and thus achieves sample-level unbiasedness that is well guaranteed under the mild assumption of $P(O_{u,i}=1 \mid R_{u,i}=1)$.
\end{itemize}

CPR loss is \emph{not} perfect --- it ensures the unbiasedness at the cost of pursuing a relaxed optimization goal: encouraging the sum score of $k$ ($k \ge 2$) positive pairs to be higher than that of $k$ negative pairs, which may not ensure the expected one-to-one comparison in Definition \ref{def:unbias}. But this relaxation on optimization can be empirically justified by the good performance of CPR.

\subsection{Comparison with Setwise Ranking}

Setwise Ranking \cite{cofiset,setrank,set2setrank} is another line of work that also make adaptations on BPR. But CPR is quite different from it in terms of both the motivation and the method.
\begin{itemize}
    \item \emph{Motivation.} Setwise Ranking aims to improve the performance of BPR by ranking item sets instead of instances. It is not designed for debiasing purpose.
    \item \emph{Method.} For Setwise Ranking, given a user $u$, two sets are constructed: $\Set{S}_u^+$ is a set of some positive items and $\Set{S}_u^-$ is a set of some negative items. Then the relevance between $\Set{S}_u^+$ and $u$ is encouraged to be closer than that of $\Set{S}_u^-$ and $u$. In other words, they still use triplets $\{(u,\Set{S}_u^+,\Set{S}_u^-)\}$ for pairwise ranking. But CPR uses a more complex form of sample $\{(u_1,\dots,u_k,i_1,\dots,i_k) \mid Y_{u_1, i_1}=1,Y_{u_2, i_2}=1,\dots,\allowbreak Y_{u_k, i_k}=1, Y_{u_1, i_2}=0,Y_{u_2, i_3}=0,\dots,Y_{u_k, i_1}=0 \}$ for cross pairwise ranking, which endows CPR with the debiasing ability that Setwise Ranking methods do not have.
\end{itemize}

\section{Experiments}

We first describe experimental settings, and then evaluate our proposed method as compared to other competitive methods. We also examine the debiasing effect and generalization ability of CPR, and study the effects of different compositions of CPR samples.

\subsection{Experimental Settings}

\subsubsection{Datasets}

\begin{table}
    \caption{Data statistics.}
    \vspace{-8pt}
    \label{tab:datasets}
    \resizebox{0.8\linewidth}{!}{
        \begin{tabular}{ccccc}
            \toprule
            Dataset   & \#Users & \#Items & \#Interactions & Density \\
            \midrule
            MovieLens & 61,770  & 6,958   & 1,533,956      & 0.0036  \\
            Netflix   & 46,420  & 12,898  & 2,475,020      & 0.0041  \\
            iFashion  & 264,251 & 53,033  & 1,505,141      & 0.0001  \\
            \bottomrule
        \end{tabular}}
    \vspace{-15pt}
\end{table}

We conduct empirical studies on three real-world datasets: MovieLens-10M~\cite{movielens}, Netflix Prize~\cite{netflix} and Alibaba iFashion~\cite{ifashion}. Ratings in MovieLens and Netflix are binarized by setting five-star ratings to 1 and the rest to 0. For iFashion, the user clicks on fashion outfits are regarded as interactions.
For all datasets, We adopt 3-core setting to keep more unpopular items instead of the commonly used 10-core setting that discards users and items with less than 10 interactions~\cite{bpr}, since the debiasing performance largely depends the performance on unpopular items.
The statistics of the processed datasets are summarized in Table \ref{tab:datasets}.

\subsubsection{Evaluation Protocol}

To evaluate the proposed debiasing approach, we follow the offline evaluation protocol in previous work~\cite{causal_inference,caus_e,dice} to create simulated unbiased data. More specifically, we sample the records from the full dataset with equal probability in terms of items to create the validation and testing sets, leaving the rest as training data. Moreover, we cap the sampling probability with an upper limit, otherwise, the probability can be large for some unpopular items, making these items over sampled into the validation/testing set, while leaving few in the training set.
More formally, let $p_s(u,i)$ be the probability of sampling the user-item pair $(u,i)$, we set $p_s(u,i) \propto \min (\frac{1}{d_i},a)$, where $d_i$ is the number of occurrences of $i$ in the full dataset, and the upper limit $a$ is set to $\frac{1}{60}$ for MovieLens and Netflix, and set to $\frac{1}{12}$ for the more sparse dataset iFashion.
The sampled data is randomly split into a validation set and a test set, and the rest data is treated as the training set. We obtain a 70/10/20 split for the training/validation/testing sets.

We measure the performance by three evaluation metrics: Recall@K, NDCG@K and ARP@K, where K is set to 20 by default.
Recall@K and NDCG@K are widely used to measure the quality of the recommendations; since here we adopt an unbiased evaluation protocol, these metrics are calculated on the simulated unbiased sets where popular items are downsampled, emphasizing the importance of unpopular items, so higher values of these metrics reflect both higher quality and better unbiasedness of the recommendation.
Average Recommendation Popularity at K (ARP@K) introduced by \cite{managing,challenging} is an additional measure for (un)biasedness. It calculates the average popularity of the top-K recommended items for each user, where the item popularity is represented by its degree. Lower values of ARP@K reflect better unbiasedness of the recommendation.

\subsubsection{Baselines}

We refer to the approaches that do not explicitly handle popularity bias as traditional approaches, and those that handle the bias as debiasing approaches. We compare the proposed CPR with two traditional approaches (the first two below) and four debiasing approaches (the rest four):
\begin{itemize}
    \item BPR~\cite{bpr}: BPR is a commonly used training method, which optimizes the personalized ranking by minimizing the pairwise ranking loss.
    \item Mult-VAE~\cite{multvae}: This is a recommendation model based on the variational autoencoder (VAE), which assumes the data is generated from a multinomial distribution and uses variational inference for parameter estimation.
    \item CausE~\cite{caus_e}: This is a domain adaptation method, which assigns two embeddings to each item --- one is learned on a large amount of biased data and the other is learned on a small amount of unbiased data. It further regularizes them to be similar.
    \item Rel-MF \cite{relmf}: It is an IPS method that modifies the traditional pointwise loss.
    \item UBPR \cite{ubpr}: It is an IPS method that modifies the traditional pairwise loss. It uses a non-negative loss function to reduce the variance of the estimator.
    \item DICE~\cite{dice}: It is the state-of-the-art debiasing method. For each user and item, it learns two disentangled embeddings for two causes, user interest and conformity, by sampling cause-specific data for training.
\end{itemize}

Except for Mult-VAE, which must use the VAE architecture and loss, other approaches, as well as our CPR can be applied across different backbones. For a fair comparison, we adopt the simple and widely used model, Matrix Factorization (MF) \cite{mf}, as the backbone of these approaches to conduct all the experiments expect for those in Section~\ref{sec:backbone}, where the aim is to compare the approaches under different backbones including LightGCN \cite{lightgcn}, NeuMF \cite{neumf}, and NGCF \cite{ngcf}.

\subsubsection{hyperparameter Settings}
We implement CPR and baselines in TensorFlow. The embedding size is 128 for all methods except for DICE. To keep the number of embedding parameters consistent among all the methods, we set the embedding size to 64 for DICE since it learns two embeddings for each user and item.
All the models are trained with Adam optimizer via early stopping.
For CPR, samples are drawn from both $\Set{D}_2$ and $\Set{D}_3$ to compute the loss (set $k=2,3$ in Equation~(\ref{cpr_loss})) and the sampling ratio of $\Set{D}_2$ and $\Set{D}_3$ is tuned on the validation set; we will empirically justify this sample composition and provide hyperparameter study on the sampling ratio in Section~\ref{sec:composition}. The dynamic sampling rate in CPR is also tuned on the validation set; the choosing rate is fixed to 2.

\subsubsection{Significance Test}

To detect significant differences between CPR and the best baseline on each dataset, we repeat their experiments 5 times by changing the random seed. The two-tailed pairwise t-test is performed using their results.

\subsection{Overall Comparison}\label{sec:overall}

Experimental results demonstrate the superiority of CPR in terms of both effectiveness and efficiency.

\subsubsection{Effectiveness Comparison}

The effectiveness comparison is presented in Table \ref{tab:mf}. Our methods are CPR-rand and CPR. CPR-rand adopts random sampling while CPR performs dynamic sampling.
CPR-rand beats all the baselines on MovieLens and Netflix. On iFashion, CPR-rand achieves similar Recall with the best baseline, but with a lower ARP which indicates that its recommendation is less biased. CPR further improves the performance of CPR-rand on all datasets, significantly outperforming all baselines on Recall and NDCG, with $p$-values less than 0.05, which indicates that the improvement achieved by CPR over the best baseline is statistically significant.
Note that the Recall and NDCG are measured on the simulated unbiased test sets where popular items are downsampled, it is normal to find their values relatively low compared to those in the papers where the test sets are randomly sampled~\cite{multvae}, because models generally perform worse on unpopular items.

The higher Recall and NDCG of CPR compared to other approaches on unbiased test set indicates its superior prediction accuracy, especially on unpopular items.
The lower ARP of CPR further indicates that CPR tends to recommend more unpopular items, in other words, it can better mitigate the popularity bias.

\begin{table*}
    \caption{Overall Performance. The best result in each column is in bold; the best baseline result in each column is underlined. Significance test is conducted on Recall and NDCG between CPR and the best baseline. $\star: p\text{-value}<0.05$. \%Improv.: percentage of improvement on Recall and NDCG of CPR over the best baseline. Higher Recall and NDCG means higher prediction accuracy; lower ARP means less popularity bias.}
    \label{tab:mf}
    \vspace{-10pt}
    \begin{tabular}{lllllllllll}
        \toprule
        \multirow{2}{*}{Method} & \multirow{2}{*}{Backbone}                & \multicolumn{3}{c}{MovieLens}            & \multicolumn{3}{c}{Netflix}              & \multicolumn{3}{c}{iFashion}                                                                                                                                                                                         \\
        \cmidrule(lr){3-5} \cmidrule(lr){6-8} \cmidrule(lr){9-11}
                                &                                          & Recall                                   & NDCG                                     & ARP                          & Recall                                   & NDCG                                     & ARP              & Recall                                   & NDCG   & ARP                      \\
        \midrule
        BPR                     & MF                                       & 0.1579                                   & 0.0939                                   & 4084                         & 0.1255                                   & 0.0952                                   & 3341             & 0.0278                                   & 0.0122 & 436                      \\
        Mult-VAE                & -                                        & 0.1676                                   & 0.1004                                   & 4468                         & 0.1242                                   & 0.0921                                   & 3554             & \underline{0.0309}
                                & \underline{0.0142}
                                & 415
        \\
        \cmidrule{1-11}
        \multicolumn{11}{l}{\emph{Debiasing approaches:}}                                                                                                                                                                                                                                                                                                                               \\
        \cmidrule{1-11}
        CausE                   & \multirow{4}{*}{MF}                      & 0.1440                                   & 0.0805                                   & 3814                         & 0.1080                                   & 0.0730                                   & 3206             & 0.0214                                   & 0.0094 & 408                      \\
        Rel-MF                  &                                          & 0.1463                                   & 0.0829                                   & 4107                         & 0.1138                                   & 0.0765                                   & 3254             & 0.0250                                   & 0.0113 & 421                      \\
        UBPR                    &                                          & 0.1682                                   & 0.1000                                   & 2691                         & 0.1315                                   & 0.1007                                   & 2289             & 0.0281                                   & 0.0126 & 404                      \\
        DICE                    &                                          & \underline{0.1835}                       & \underline{0.1101}                       & \underline{1173}             & \underline{0.1317}                       & \underline{0.1007}                       & \underline{1244} & 0.0275                                   & 0.0124 & \textbf{\underline{250}} \\
        \cmidrule{1-11}
        CPR-rand                & \multirow{2}{*}{MF}                      & 0.1938                                   & 0.1192                                   & \textbf{1055}                & 0.1462                                   & 0.1143                                   & 1270             & 0.0307
                                & 0.0137
                                & 398                                                                                                                                                                                                                                                                                                                                                   \\
        CPR                     &                                          & \textbf{0.2003}\textsuperscript{$\star$} & \textbf{0.1223}\textsuperscript{$\star$} & 1138                         & \textbf{0.1511}\textsuperscript{$\star$} & \textbf{0.1190}\textsuperscript{$\star$} & \textbf{1204}    & \textbf{0.0332}\textsuperscript{$\star$}
                                & \textbf{0.0151}\textsuperscript{$\star$}
                                & 359
        \\
        \cmidrule{1-11}
        \%Improv.               &                                          & 9.16\%                                   & 11.08\%                                  & -                            & 14.73\%                                  & 18.17\%                                  & -                & 7.44\%
                                & 6.34\%                                   & -
        \\
        \bottomrule
    \end{tabular}
    \vspace{-5pt}
\end{table*}

\subsubsection{Efficiency Comparison}

Figure \ref{fig:curves} shows the Recall curves of all approaches on MovieLens and Netflix, where it can be observed that CPR converges to the best performance with the least number of epochs, contributing by the dynamic sampling. Although dynamic sampling processes more samples to select the relatively difficult ones, it can be implemented efficiently with negligible cost.
Figure \ref{fig:time} compares the total training time of all approaches; the y-axis is log-scaled. The strongest baseline DICE requires the longest training time, while our CPR is relatively fast among all approaches. More specifically, CPR takes only 2.6\% and 1.3\% of the training time of DICE on MovieLens and Netflix, respectively.

\begin{figure}
    \vspace{-5pt}
    \begin{subfigure}{.5\linewidth}
        \includegraphics[width=\linewidth]{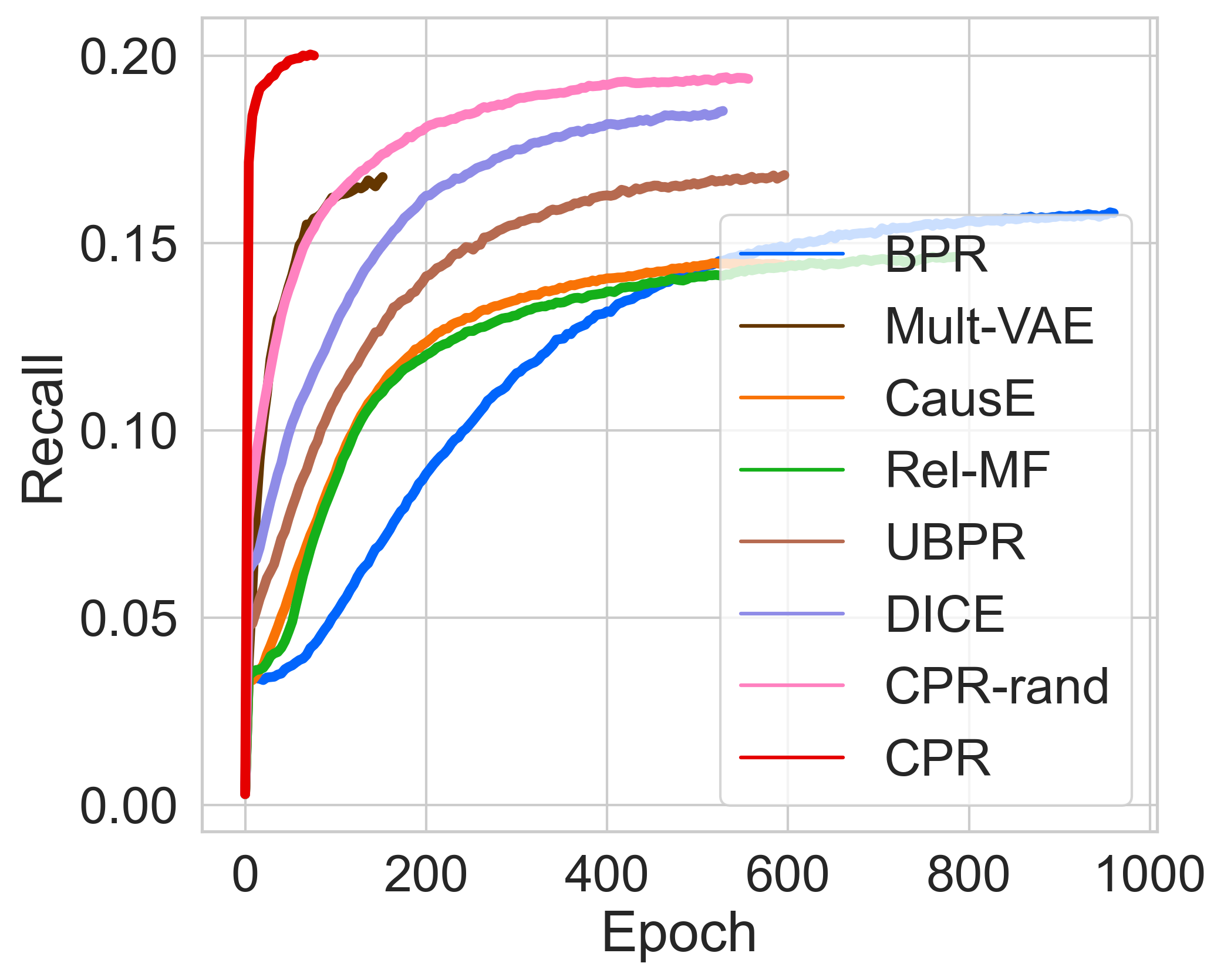}
        \caption{MovieLens}
        \vspace{-8pt}
    \end{subfigure}%
    \begin{subfigure}{.5\linewidth}
        \includegraphics[width=\linewidth]{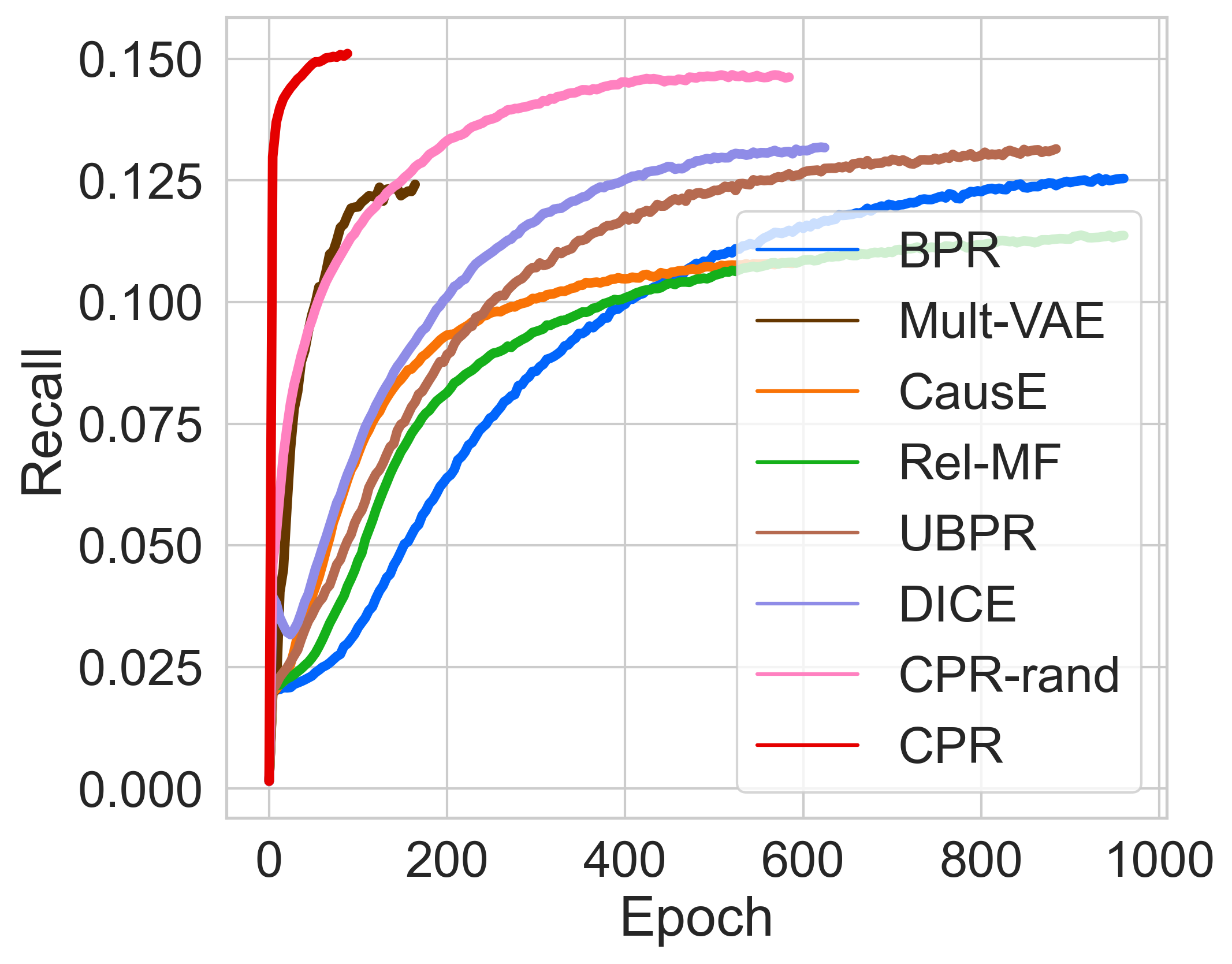}
        \caption{Netflix}
        \vspace{-8pt}
    \end{subfigure}
    \caption{Comparison of Recall curves.}
    \label{fig:curves}
    \vspace{-10pt}
\end{figure}

\begin{figure}
    \centering
    \begin{subfigure}{.45\linewidth}
        \includegraphics[width=\linewidth]{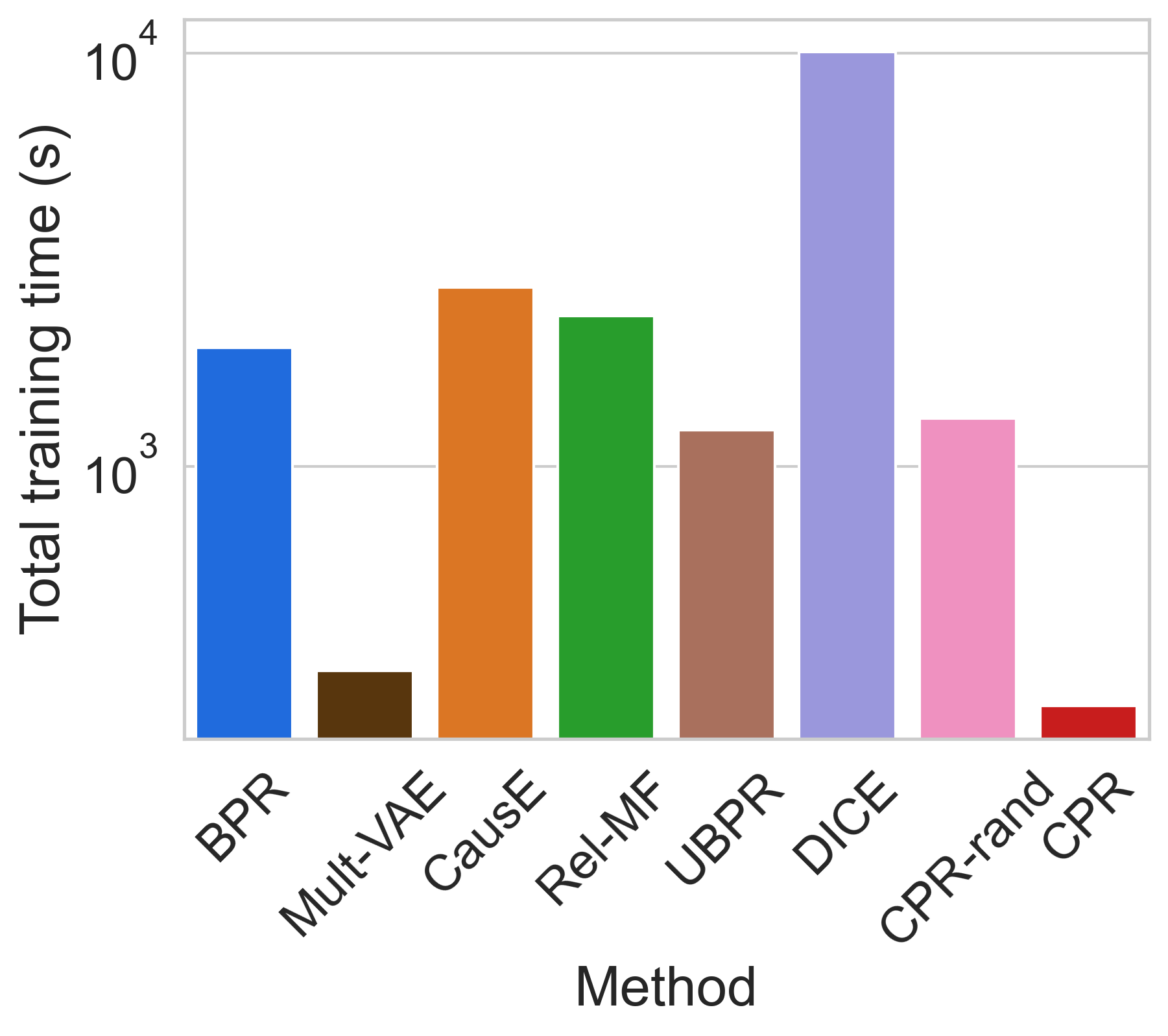}
        \caption{MovieLens}
        \vspace{-10pt}
    \end{subfigure}%
    \begin{subfigure}{.45\linewidth}
        \includegraphics[width=\linewidth]{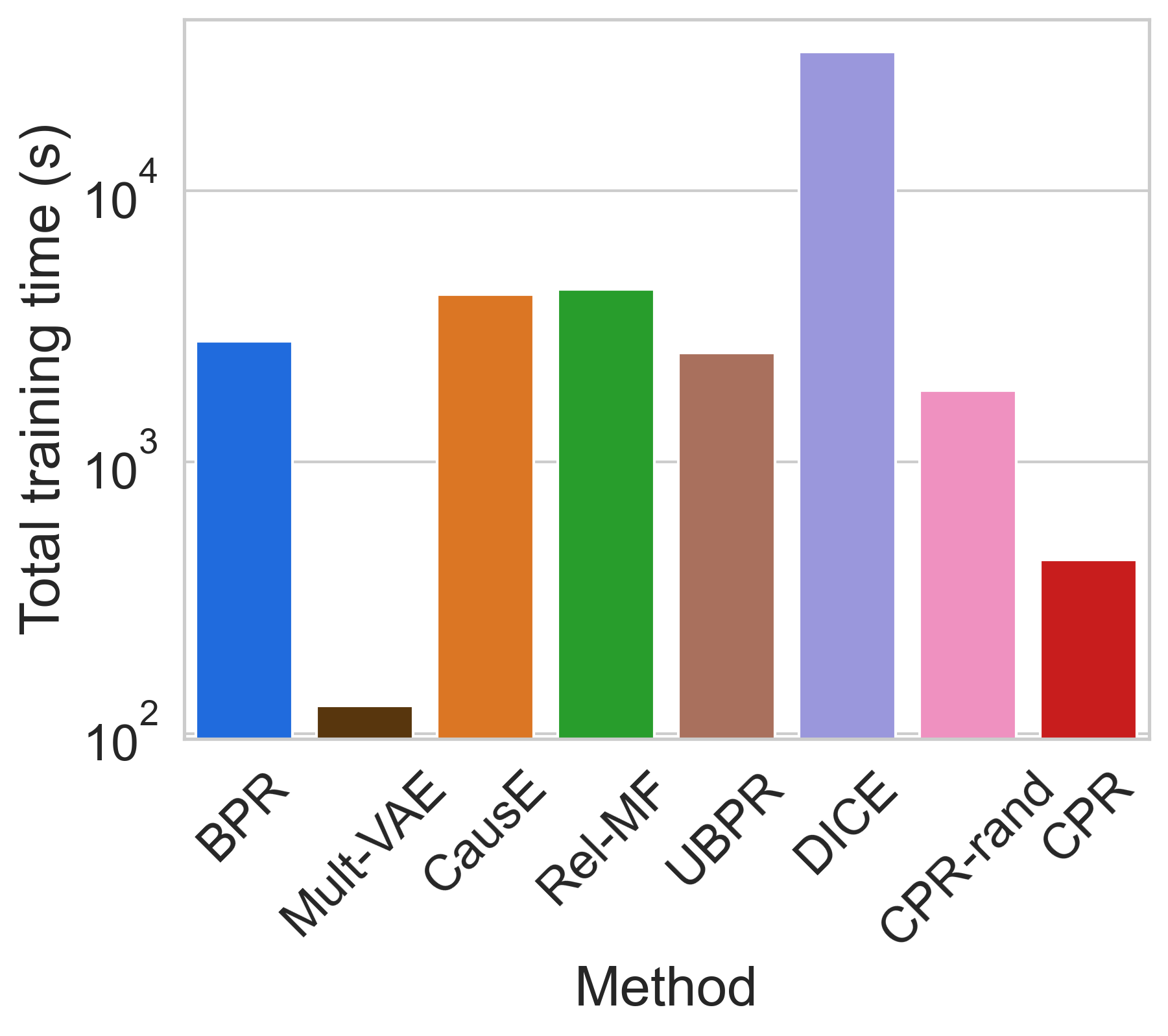}
        \caption{Netflix}
        \vspace{-10pt}
    \end{subfigure}
    \caption{Comparison of the total training time.}
    \label{fig:time}
    \vspace{-10pt}
\end{figure}

\subsection{Distribution of Recommendation}\label{sec:distribution}

In order to examine the debiasing effect of CPR more intuitively, we compare the percentage of recommended items grouped by their degrees for each approach, along with their percentage in the training and test sets.

We categorize all items into four groups in the following way: sorting all items by their degrees in the ascending order, and grouping them such that the sum of degrees in each group are approximately equal. Figure \ref{fig:rec} shows the percentage of degrees contained in each group of the training set and the test set, and the percentage of the recommended items of each group by each approach. The curve of the training set is flat due to our grouping method. While for the testing set, the percentage of recommended items decreases as the group ID increases since popular items are downsampled, as a result, low-degree items take a larger proportion. For an ideal unbiased recommendation, its distribution is expected to match that of the test set.

BPR and Mult-VAE both amplify the bias in the training set, recommending more popular items from Group 2 and 3. Debiasing approaches like UBPR, DICE and CPR alleviate the bias by recommending more unpopular items from Group 0 and 1. CPR and DICE achieve the best debiasing result, making the distribution closer to that of the test set.

\begin{figure}
    \begin{subfigure}{.5\linewidth}
        \centering
        \includegraphics[width=\linewidth]{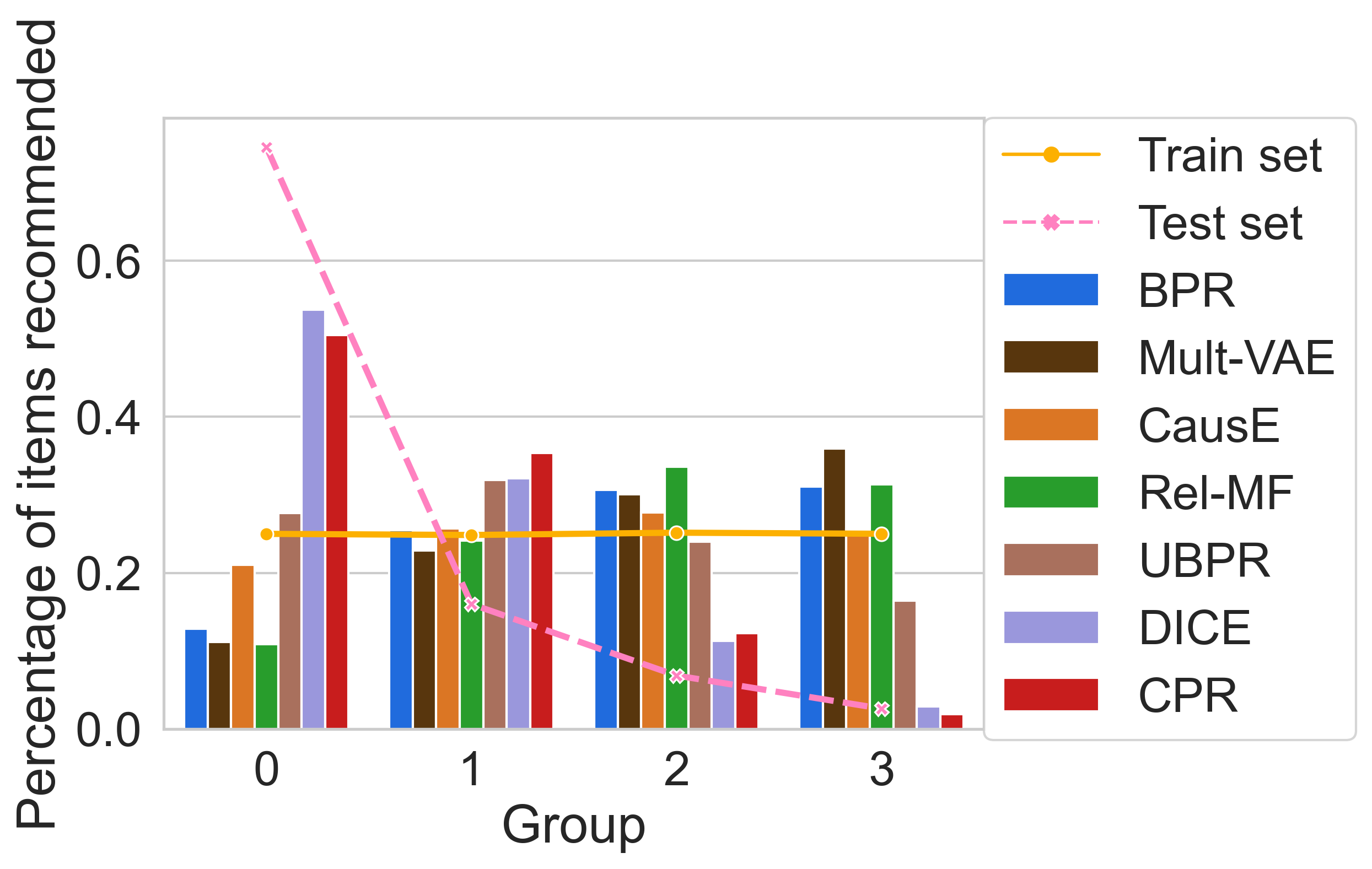}
        \vspace{-10pt}
        \caption{MovieLens}
    \end{subfigure}%
    \begin{subfigure}{.5\linewidth}
        \centering
        \includegraphics[width=\linewidth]{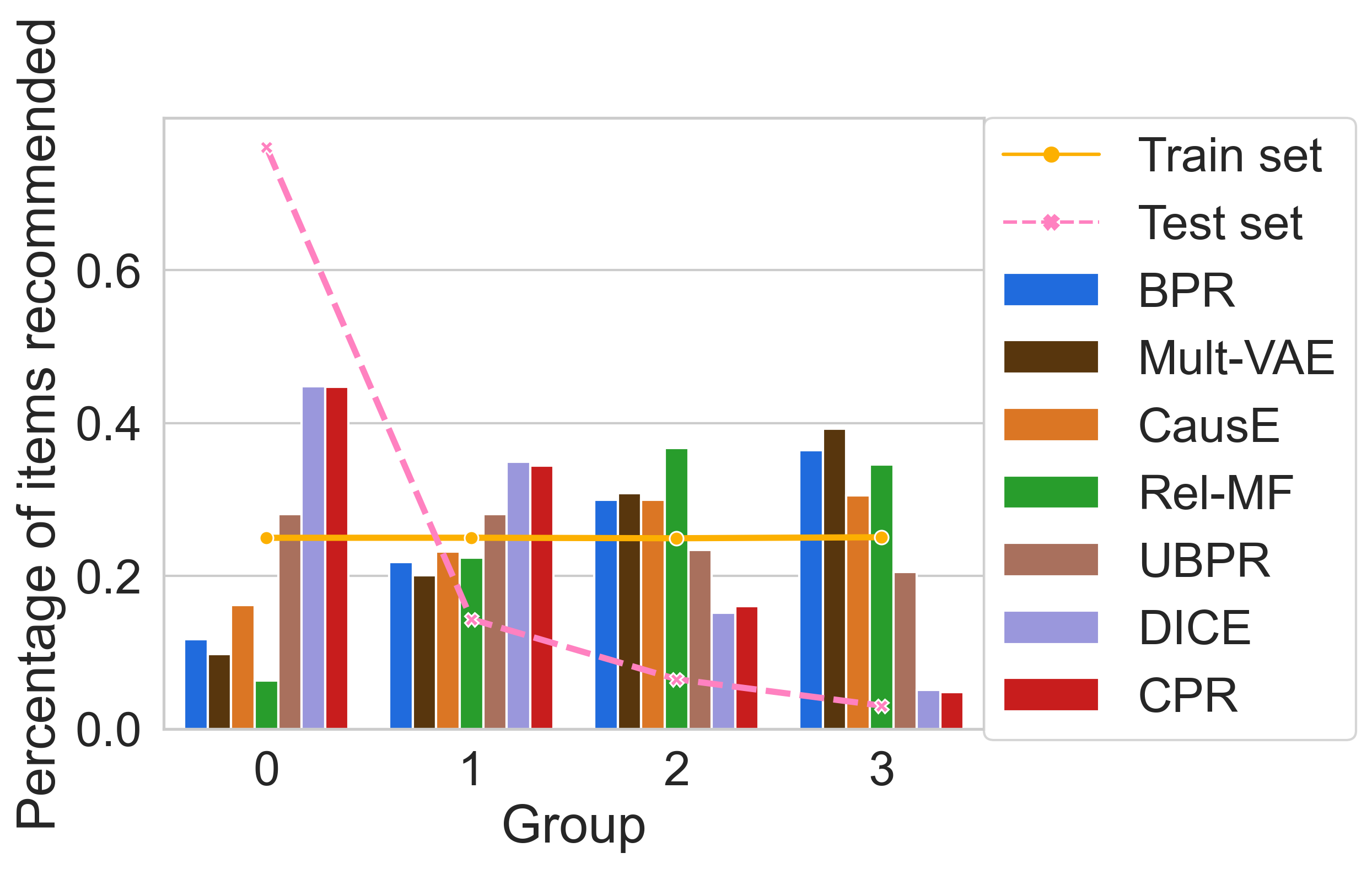}
        \vspace{-10pt}
        \caption{Netflix}
    \end{subfigure}
    \caption{Comparison of percentage of recommended items of each group.}
    \label{fig:rec}
    \vspace{-10pt}
\end{figure}

\subsection{Generalization Ability}\label{sec:generalization}

\subsubsection{Different Backbones}\label{sec:backbone}

We apply our method CPR, as well as BPR and the best baseline DICE on MovieLens and Netflix, to different backbones to evaluate their generalization ability. We adopt the state-of-the-art graph model LightGCN \cite{lightgcn}, a neural model NeuMF \cite{neumf}, and another graph model NGCF \cite{ngcf} as different backbones. Results in Table \ref{tab:backbones} verify that CPR consistently outperforms BPR and DICE on these backbones. For example, on MovieLens, CPR performs better than DICE by 7.9\%, 26.4\% and 10.4\% \wrt NDCG@20 on LightGCN, NeuMF and NGCF, respectively.

\begin{table}
    \caption{Performance \wrt different backbones.}
    \vspace{-8pt}
    \label{tab:backbones}
    \resizebox{0.9\linewidth}{!}{
        \begin{tabular}{llllll}
            \toprule
            \multirow{2}{*}{Method} & \multirow{2}{*}{Backbone} & \multicolumn{2}{c}{MovieLens} & \multicolumn{2}{c}{Netflix}                                     \\
            \cmidrule(lr){3-4}\cmidrule(lr){5-6}
                                    &                           & Recall                        & NDCG                        & Recall          & NDCG            \\
            \midrule
            BPR                     & \multirow{3}{*}{LightGCN} & 0.1661                        & 0.0995                      & 0.1310          & 0.0985          \\
            DICE                    &                           & 0.1864                        & 0.1121                      & 0.1355          & 0.1043          \\
            CPR                     &                           & \textbf{0.1963}               & \textbf{0.1210}             & \textbf{0.1502} & \textbf{0.1180} \\
            \midrule
            BPR                     & \multirow{3}{*}{NeuMF}    & 0.1519                        & 0.0890                      & 0.1254          & 0.0932          \\
            DICE                    &                           & 0.1568                        & 0.0914                      & 0.1126          & 0.0841          \\
            CPR                     &                           & \textbf{0.1895}               & \textbf{0.1155}             & \textbf{0.1434} & \textbf{0.1109} \\
            \midrule
            BPR                     & \multirow{3}{*}{NGCF}     & 0.1615                        & 0.0951                      & 0.1199          & 0.0899          \\
            DICE                    &                           & 0.1774                        & 0.1071                      & 0.1312          & 0.0999          \\
            CPR                     &                           & \textbf{0.1952}               & \textbf{0.1182}             & \textbf{0.1498} & \textbf{0.1173} \\
            \bottomrule
        \end{tabular}}
    \vspace{-10pt}
\end{table}

\subsubsection{Different Degrees of Bias}

We also evaluate the generalization ability of CPR using training data with different degrees of bias. We sample 70\% of the training set with probability $p_s(u,i) \propto d_i^\theta$, where $d_i$ is the item degree in the training set, and $\theta$ is set to 0.5 and -0.5, respectively. When $\theta=0.5$, the bias in the sampled training set is amplified compared with the original training set; when $\theta=-0.5$, the bias is reduced. Table~\ref{tab:train} shows that CPR consistently outperforms the best baselines, DICE, with around 10\% and 14\% improvement \wrt NDCG@20 on MovieLens and Netflix, respectively with both sampled training sets ($\theta=0.5$ and $\theta=-0.5$).

\begin{table}
    \caption{Performance \wrt different degrees of bias.}
    \vspace{-8pt}
    \label{tab:train}
    \resizebox{0.8\linewidth}{!}{
        \begin{tabular}{llllll}
            \toprule
            \multirow{2}{*}{$\theta$} & \multirow{2}{*}{Method} & \multicolumn{2}{c}{MovieLens} & \multicolumn{2}{c}{Netflix}                                     \\
            \cmidrule(r){3-4}\cmidrule(r){5-6}
                                      &                         & Recall                        & NDCG                        & Recall          & NDCG            \\
            \midrule
            \multirow{3}{*}{0.5}      & BPR                     & 0.1217                        & 0.0679                      & 0.0939          & 0.0654          \\
                                      & DICE                    & 0.1461                        & 0.0839                      & 0.1048          & 0.0741          \\
                                      & CPR                     & \textbf{0.1584}               & \textbf{0.0927}             & \textbf{0.1189} & \textbf{0.0858} \\
            \midrule
            \multirow{3}{*}{-0.5}     & BPR                     & 0.1413                        & 0.0809                      & 0.1086          & 0.0777          \\
                                      & DICE                    & 0.1597                        & 0.0956                      & 0.1133          & 0.0858          \\
                                      & CPR                     & \textbf{0.1737}               & \textbf{0.1040}             & \textbf{0.1274} & \textbf{0.0965} \\
            \bottomrule
        \end{tabular}}
\end{table}

\subsection{Composition of Samples}\label{sec:composition}

We compare the performances of CPR with different sample compositions, \ie different $k$ values in Equation~(\ref{cpr_loss}). We first fix $k$ to single values: $k=1$, $k=2$, or $k=3$ and tune other hyperparameters to obtain the best performances under each value of $k$. The dashed lines in Figure~\ref{fig:composition} show that the value of Recall decreases as $k$ increases, but not much --- the performance under $k=4$ is still better than all baselines in Table~\ref{tab:mf}. As we discussed in Section~\ref{sec:extend}, the performance drop is probably due to the inflexibility of training that comes with the larger $k$. Next, we set $k=2,3$. We choose the sampling ratio of $\Set{D}_2$ and $\Set{D}_3$ from $\left[1,2,3,4,5,6\right]$, and tune other hyperparameters under each value of the sampling ratio. Results are drawn with solid lines in Figure~\ref{fig:composition}. The performance under $k=2,3$ can be better than that under $k=2$ with the appropriate sampling ratio which is typically around 3, indicating that different types of samples can be used in combination to achieve better results compared to using a single type of samples.

\begin{figure}
    \begin{subfigure}{.5\linewidth}
        \includegraphics[width=\linewidth]{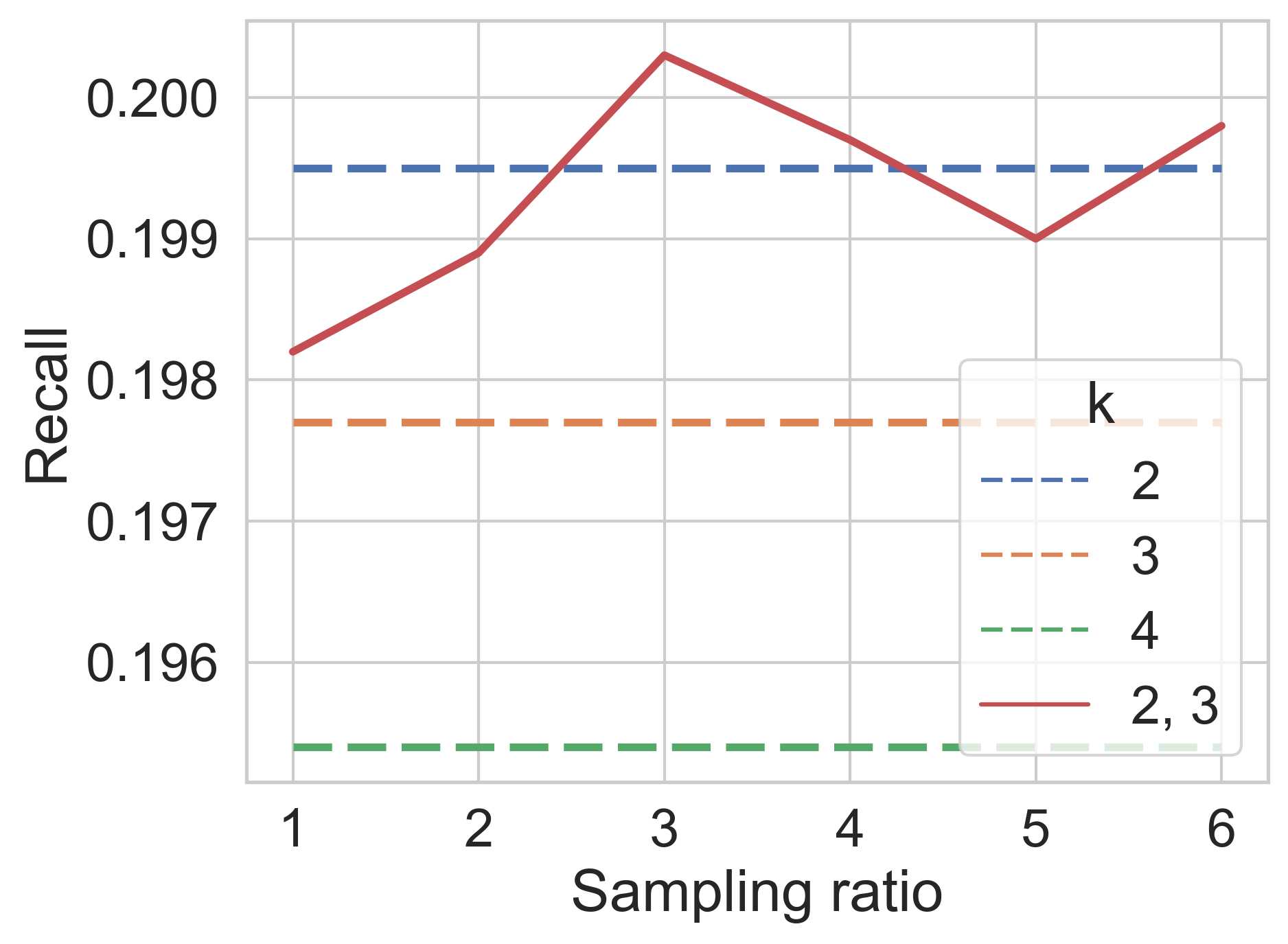}
        \caption{MovieLens}
        \vspace{-8pt}
    \end{subfigure}%
    \begin{subfigure}{.5\linewidth}
        \includegraphics[width=\linewidth]{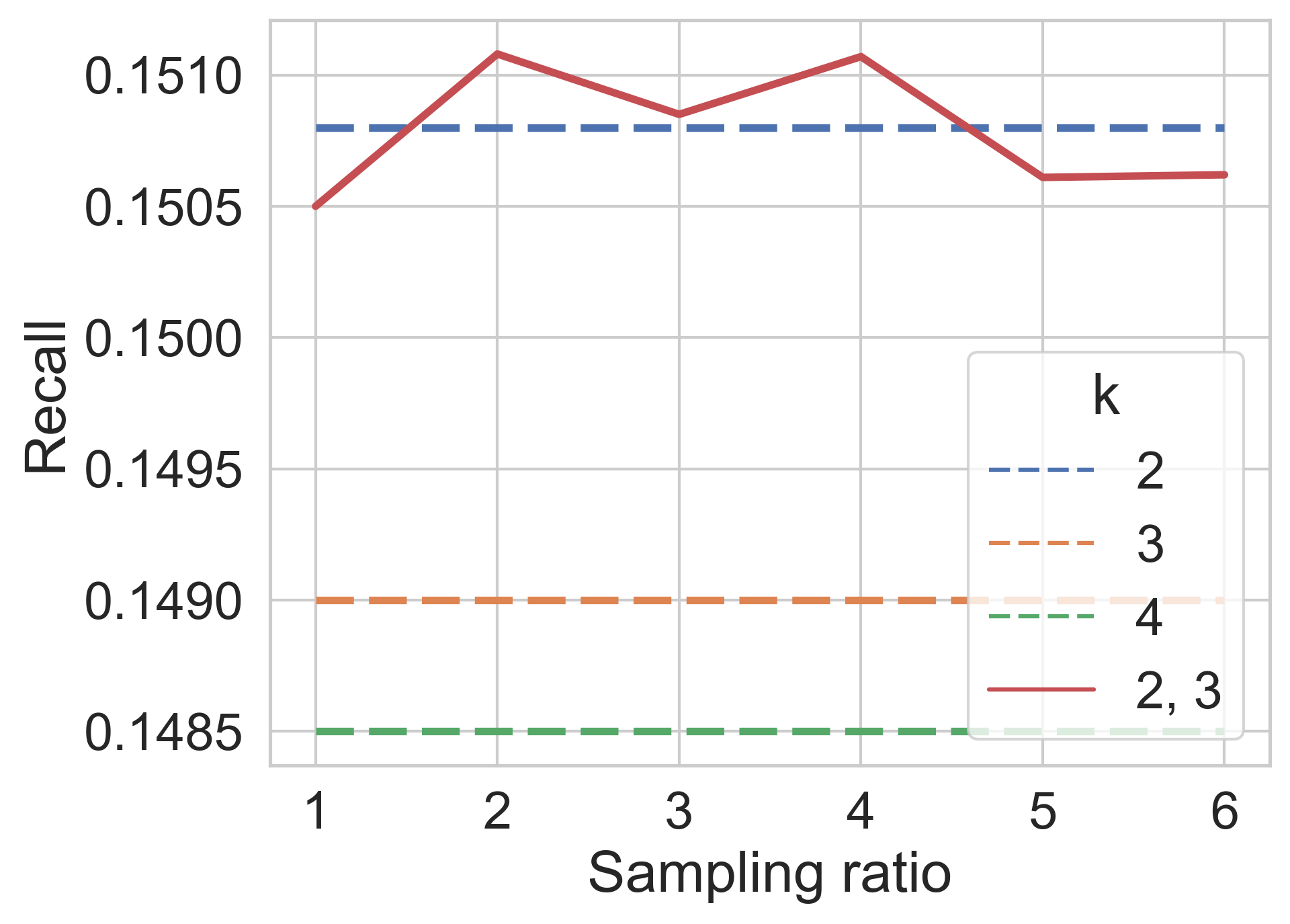}
        \caption{Netflix}
        \vspace{-8pt}
    \end{subfigure}
    \caption{Comparison of sample compositions.}
    \label{fig:composition}
    \vspace{-8pt}
\end{figure}
\section{Related Work}

Generally, there are three research lines for unbiased recommendation on implicit feedback. Here we review each of them.

\subsubsection{Domain Adaptation}
Some studies perform domain adaptation by utilizing a small amount of unbiased data as the target domain to guide the training on biased source data.
For example, CausE~\cite{caus_e} corrects the biased embedding by minimizing its distance from another embedding learned on the unbiased data. KDCRec~\cite{kdc_rec} uses knowledge distillation to transfer the knowledge obtained from biased data to the modeling of unbiased data. However, since the unbiased data is often rare, their performance is limited in real scenarios.
To avoid the usage of unbiased data, ESAM~\cite{esam} treats displayed and non-displayed items as source and target domains, respectively, and minimizes the distance between their covariance matrices.
However, domain adaptation methods often suffer from loss of domain-specific information --- some important item features learned from biased data may be lost when aligning with the features from unbiased data.

\subsubsection{Counterfactual Learning}
Counterfactual learning methods utilize causal graphs to handle the effect of item popularity on recommendation.
DICE~\cite{dice} uses two separate embeddings for each user and item to model the causal effects of user-item relevance and exposure mechanism, respectively, by training with cause-specific data.
MACR~\cite{macr} and CR~\cite{cr} removes the direct effect of item properties on predicted scores by causal inference.
DIB~\cite{dib} and PDA~\cite{pda} both remove the confounding popularity bias during training, but PDA~\cite{pda} further inject the future popularity to the scores during inference.
Unlike our proposed method that focuses on constructing an unbiased loss function, these approaches deal with the popularity bias from a different point of view --- they analyze the causal effect between the bias and the observed data and then apply causal operations accordingly.

\subsubsection{Inverse Propensity Score (IPS)}
IPS is a prevalent debiasing method for recommender systems \cite{rec_treat,causal_inference}, and recently explored in the implicit recommendation setting. Rel-MF~\cite{relmf} is the first IPS method designed for implicit feedback, which reweights the pointwise loss by the estimated item exposure probability to obtain an expectation-unbiased loss function. Later, UBPR~\cite{ubpr} is proposed to extend the pointwise model in Rel-MF to a pairwise version.
UEBPR~\cite{uebpr} and DU~\cite{du} introduce neighborhood-based explainability and unclicked data reweighting to the plain IPS methods, respectively.
AutoDebias~\cite{autodebias} combines IPS with data imputation and adopts a meta-learning algorithm to learn the optimal debiasing configurations on a small uniform data.
These IPS methods mainly have two drawbacks: the difficulty in estimating propensity scores and the high variance of the reweighted loss.
\section{Conclusion and Future Work}

In this work, we revisit the commonly used pointwise and pairwise loss functions from a new perspective, and point out that they are biased in approaching the correct ranking of user preference.
Under the assumption that the exposure probability can be factorized as user propensity, item propensity, and user-item relevance, we ascribe the biasedness in traditional loss functions to user/item propensity.
We then propose a new loss, CPR, where the impact of the user/item propensity is offset by the combination of carefully selected user-item pairs, so as to approach the unbiased ranking of user preference.
Experimental results show that CPR achieves better debiasing performance on multiple backbone models.

This initial work sheds a light on mitigating the popularity bias without the knowledge of the exposure mechanism, with the help of a mild assumption.
One limitation of CPR is that more interactions in one sample would bring inflexibility to training, which prevents us from using more types of samples together to further improve the recommendation accuracy. We hope to solve this problem in future work.
Another limitation is that the proposed assumption of the exposure probability may not reflect the real exposure mechanism accurately.
We hope to propose a more general assumption in future work.
Moreover, we would like to extend CPR to other scenarios, such as exploring its ability to alleviate the group fairness bias \cite{fairness_pairwise}, combining it with list ranking \cite{cpr_online} or graph learning \cite{refine,sgl,dir}.

\begin{acks}
  This work is supported by the National Key Research and Development Program of China (2020AAA0106000), and the National Natural Science Foundation of China (U19A2079, 62121002).
\end{acks}

\bibliographystyle{ACM-Reference-Format}
\bibliography{cpr-ref}

\end{document}